\documentclass[a4paper]{article} 

\usepackage{fancyhdr} 
\usepackage{charter} 
\usepackage[utf8]{inputenc} 
\usepackage{microtype} 

\usepackage{amsmath,amssymb,amsthm,easybmat,mathtools}
\usepackage{tikz}
\usepackage{float}
\usepackage{verbatim}
\usepackage{enumerate}
\usepackage{color}
\usepackage{graphicx}
\usepackage{subfig}
\usepackage{longtable}
\usepackage{array}
\usepackage{dsfont}
\usepackage{csquotes}
\usepackage{bm}
\usepackage[numbers,sort&compress]{natbib}
\usepackage{doi}
\usepackage{hyperref}
\definecolor{darkred}  {rgb}{0.5,0,0}
\definecolor{darkblue} {rgb}{0,0,0.5}
\definecolor{darkgreen}{rgb}{0,0.5,0}
\hypersetup{
	colorlinks = true,
	urlcolor  = blue,         
	linkcolor = darkblue,     
	citecolor = darkgreen,    
	filecolor = darkred       
}

\newtheorem{lemma}{Lemma}
\newtheorem{theorem}{Theorem}
\newtheorem{proposition}[theorem]{Proposition}

\theoremstyle{remark}
\newtheorem*{remark}{Remark}

\newcommand{\mbb}{\mathbb}
\newcommand{\mc}{\mathcal}

\newcommand{\tr}{\textrm{Tr}}

\newcommand{\ket}[1]{|#1\rangle}
\newcommand{\bra}[1]{\langle #1|}
\newcommand{\op}[2]{|#1\rangle\langle #2|}

\newcommand{\state}[1]{\ket{#1}\bra{#1}}

\definecolor{cool_green}{rgb}{0.0, 0.5, 0.0}
\newcommand{\eric}[1]{{\color{cool_green} #1}}

\title{Capacities of Entanglement Distribution From a Central Source}
\author{Xinan Chen, Stefano Chessa, Ian George, Felix Leditzky, and Eric Chitambar \\
\textit{University of Illinois at Urbana-Champaign, Urbana, IL 61801, USA}}

\begin{document}
\maketitle


\begin{abstract}
    Distribution of entanglement is an essential task in quantum information processing and the realization of quantum networks. In our work, we theoretically investigate the scenario where a central source prepares an $N$-partite entangled state and transmits each entangled subsystem to one of $N$ receivers through noisy quantum channels. The receivers are then able to perform local operations assisted by unlimited classical communication to distill target entangled states from the noisy channel output. In this operational context, we define the EPR distribution capacity and the GHZ distribution capacity of a quantum channel as the largest rates at which Einstein-Podolsky-Rosen (EPR) states and Greenberger-Horne-Zeilinger (GHZ) states can be faithfully distributed through the channel, respectively. We establish lower and upper bounds on the EPR distribution capacity by connecting it with the task of assisted entanglement distillation. We also construct an explicit protocol consisting of a combination of a quantum communication code and a classical-post-processing-assisted entanglement generation code, which yields a simple achievable lower bound for generic channels.
    As applications of these results, we give an exact expression for the EPR distribution capacity over two erasure channels and bounds on the EPR distribution capacity over two generalized amplitude damping channels. We also bound the GHZ distribution capacity, which results in an exact characterization of the GHZ distribution capacity when the most noisy channel is a dephasing channel.
\end{abstract}

\section{Introduction}
\label{sec:intro}

Entanglement distribution is an essential ingredient in realizing the quantum internet \cite{kimble-2008,wehner-2018}, which has the potential of accomplishing tasks such as quantum key distribution \cite{bennett-1984,ekert-1991}, distributed quantum computing \cite{caleffi-2022}, and distributed quantum sensing \cite{zhang-2021}. To successfully implement these protocols, the participating parties, which may be separated by significant spatial distance, need to share high-fidelity entangled states. Typically, the distribution of entanglement to distant parties suffers from noise, deteriorating the quality of entanglement. Therefore, a key research thrust in quantum information science is to study protocols that allow us to distribute high-fidelity entangled states to distant parties. 

In previous works on entanglement distribution, it is frequently assumed that one party prepares an $N$-partite entangled state, holds on to one subsystem which remains intact in a local quantum memory, and sends the other $N-1$ subsystems to the remaining $N-1$ parties through quantum channels. The parties may also use unlimited local operation and classical communication (LOCC) post-processing to enhance the distribution of entanglement. This is illustrated in Fig. \ref{fig:one-sided-diagram}. For instance, in Refs. \cite{bennett1996mixed,wilde2010convolutional}, the distribution of Einstein-Podolsky-Rosen (EPR) states is accomplished by locally generating $n$ copies of EPR states $\phi^{A_1S_2}$ at $A_1$, who sends the $S_2^n$ systems through $n$ uses of a quantum channel $\mc{N}^{S_2\rightarrow A_2}$ to another party $A_2$. Subsequently, $A_1$ and $A_2$ employ either the hashing protocol or the convolutional protocol to distill high-fidelity EPR pairs between them. The convolutional method has further been extended to the distribution of tripartite Greenberger-Horne-Zeilinger (GHZ) states $\ket{\text{GHZ}_3}=1/\sqrt{2}(\ket{000}+\ket{111}) $\cite{rengaswamy2021distilling}. Similarly, this process involves locally generating $n$ copies of $\state{\text{GHZ}_3}^{A_1S_2S_3}$, sending $S_2^n$ and $S_3^n$ to $A_2$ and $A_3$ respectively, and performing a distillation protocol inspired by quantum error correction. Building on this work, Ref. \cite{rengaswamy2024entanglement} adopted quantum low-density parity-check codes \cite{breuckmann2021quantum,panteleev2022asymptotically} to construct a high-rate tripartite GHZ distribution protocol with efficient post-processing. Conversely, if we can prepare arbitrary bipartite entangled states and implement arbitrary LOCC post-processing, Tomamichel et al. provide a fundamental upper bound on the asymptotically achievable rates for distributing EPR states \cite{tomamichel-2017}. 

However, we argue that this previously adopted model is inadequate for the following two reasons. First, the assumption that the local quantum memory is noiseless is hardly justified in practice. Secondly, we envision that, in the near-term, the ability to generate and distribute quantum entanglement may be localized to a few central stations, which send the entangled states to users through quantum channels whenever they request them. Therefore, in our work, we instead consider a more practical scenario where a source prepares an $N$-partite entangled state $\rho_{\mathrm{init}}^{S_1\cdots S_N}$ and sends subsystem $S_i$ to receivers $A_i$ through quantum channels $\mc{N}_i^{S_i\rightarrow A_i}$. After receiving the noisy output state, the receivers perform LOCC amongst themselves to distill the desired target state $\tau$ from the noisy output state (Fig~\ref{fig:multi-sided-diagram}).

The focus of the present paper is to investigate the ultimate entanglement distribution rate in this setting, characterized by the best asymptotic rate achievable when the source is able to use each channel an unbounded number of times.  In this work, we choose $\tau$ to be the $N$-partite GHZ state $\ket{\mathrm{GHZ}_N}=1/\sqrt{2}(\ket{0}^{\otimes N}+\ket{1}^{\otimes N})$. While there is no ``maximally entangled" state in the multipartite setting \cite{acin-2002,walter-2016}, GHZ states are still an important class of states to consider because of their nonlocality properties \cite{greenberger-1990,mermin-1990} and their applications in quantum information tasks such as quantum secret sharing \cite{hillery-1999}, quantum conference key agreement \cite{murta-2020}, and quantum sensing \cite{degen-2017}. For these reasons, the distillation of GHZ states is an important question that is under active research \cite{rengaswamy2021distilling,rengaswamy2024entanglement,vrana-2019,salek-2022,salek-2023,Frantzeskakis2023,vairogs2024localizing}. For example, Salek and Winter recently proposed a GHZ distillation protocol by ``making coherent" the classical common randomness distillation protocol. Moreover, when $N=2$, the task reduces to distributing EPR pairs across two receivers, which is an interesting scenario in its own right. We will call the largest rate at which bipartite maximally entangled states can be distributed through channels $\mc{N}_1$ and $\mc{N}_2$ the EPR distribution capacity of $\mc{N}_1$ and $\mc{N}_2$ and denote it by $E(\mc{N}_1,\mc{N}_2)$. Similarly, the largest rate at which GHZ states can be distributed through $\mc{N}_1,\cdots,\mc{N}_N$ will be called the GHZ distribution capacity of these channels, denoted by $E(\mc{N}_1,\cdots,\mc{N}_N)$.

\begin{figure}
    \centering
    \subfloat[]{\includegraphics[width=0.49\textwidth]{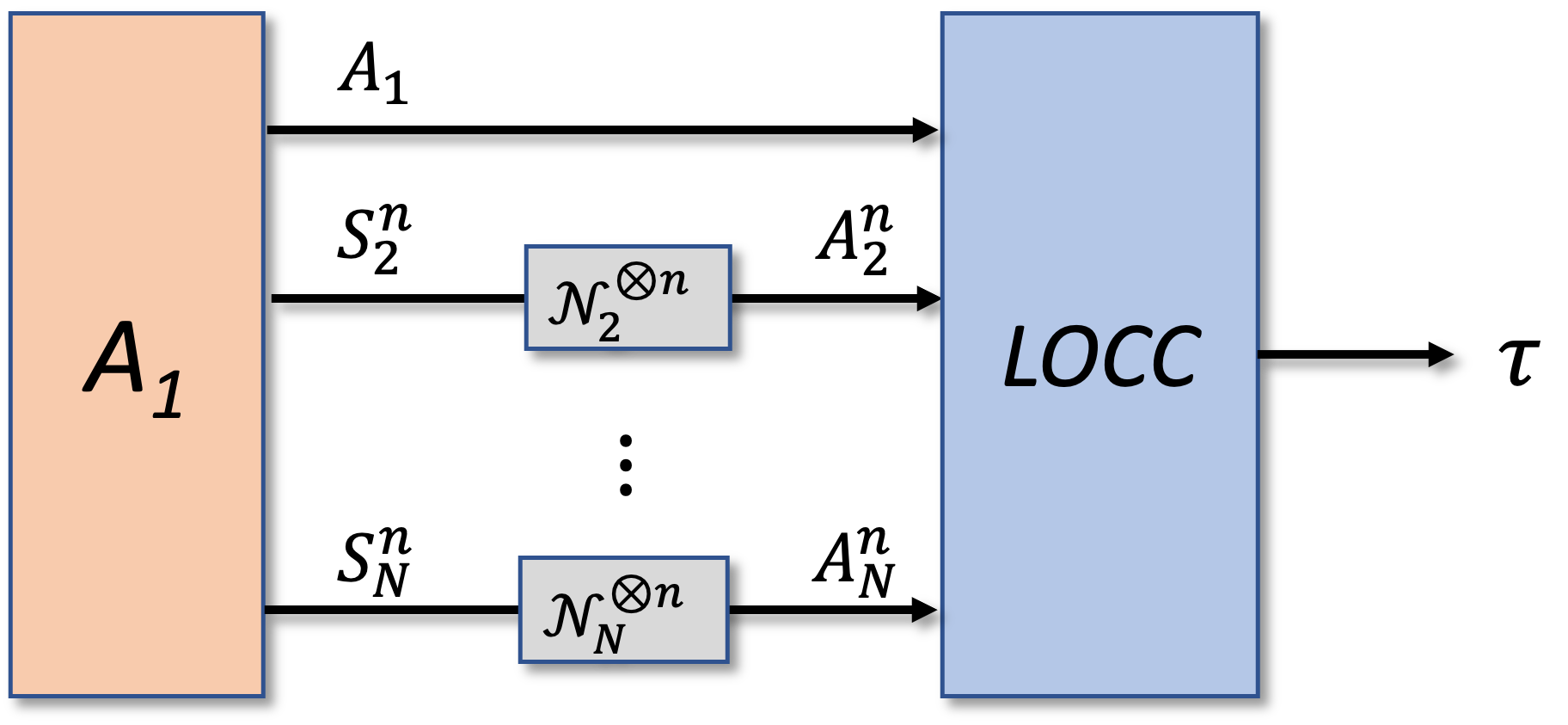}\label{fig:one-sided-diagram}}
    \hfill
    \subfloat[]{\includegraphics[width=0.49\textwidth]{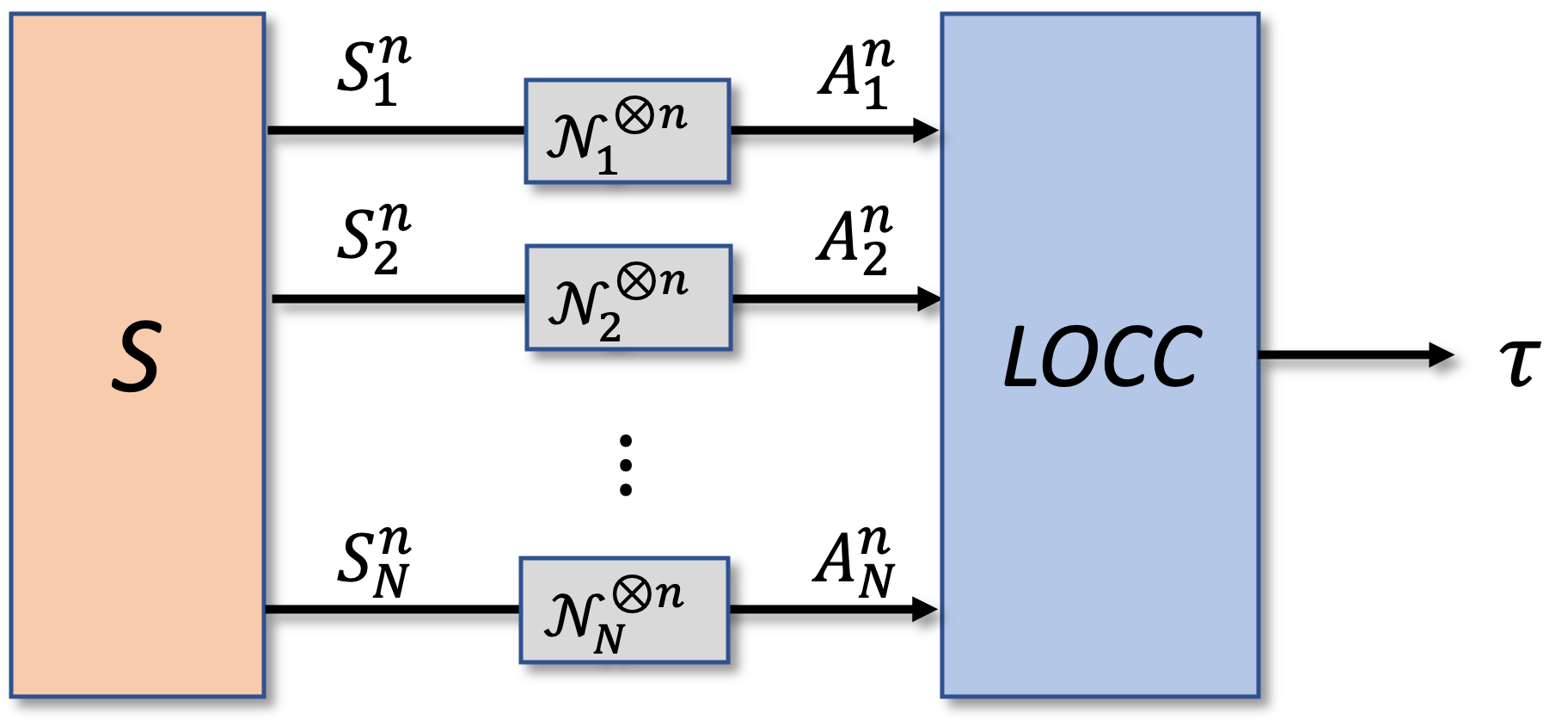}\label{fig:multi-sided-diagram}}
    \caption{Two entanglement distribution schemes. (a) One of the parties, say $A_1$, prepares an entangled state $\rho^{A_1S_2^n\cdots S_N^n}$ and sends $S_i^n$ through $n$ uses of the quantum channel $\mc{N}_i$ to $A_i$ $(i=2,3,\cdots,N)$. Then $A_i$ perform LOCC among themselves to obtain the target state $\tau$. (b) Our entanglement distribution model. The source $S$ prepares an entangled state $\rho^{S_1^n\cdots S_N^n}$, sending the $S_i^n$ subsystem to $A_i$ via $n$ uses of the channel $\mc{N}_i$. Then the receivers perform LOCC to obtain the target state $\tau$.}
\end{figure}

The main goal of this paper is to characterize the optimal $N$-party GHZ distribution rates over noisy channels, especially for the case of $N=2$. In Section \ref{sect:preliminaries}, we provide preliminary information on quantum capacities, entanglement distillation, and simulation of a quantum channel by its Choi states, all of which are useful for the exposition of our results in later sections. We also provide a detailed definition of the entanglement distribution capacity in Section \ref{sect:entanglement-distribution-definition}. Our first contribution is the connection between the task of distributing EPR states to $N=2$ parties with the task of assisted entanglement distillation \cite{divincenzo1999entanglement,smolin2005entanglement,dutil2011assisted,buscemi2012general,pollock2021entanglement} (Section \ref{sect:duality}). Specifically, we show that for general channels $\mc{N}_1$ and $\mc{N}_2$, the regularized entanglement of assistance of the Choi states offers a lower bound for $E(\mc{N}_1,\mc{N}_2)$. For teleportation-covariant channels, this lower bound is tight and exactly equals the entanglement distribution rate. As an application of this result, we exactly characterize the entanglement distribution capacity when the two channels are both erasure channels in Section \ref{sect:erasure-channels}. In Section \ref{sect:explicit-protocol}, we provide another lower bound for $E(\mc{N}_1,\mc{N}_2)$ by combining quantum communication codes with entanglement generation codes assisted by classical-post-processing (CPP). This bound is tight when the quantum capacity of one channel is greater than the CPP-assisted entanglement generation capacity of the other channel and not tight in general. In particular, for the case of EPR distribution over generalized amplitude damping channels (GADCs), we can enhance the achievable rates in the high-noise regime by using multi-rail encoding (Section \ref{sect:GADC}). In Section \ref{sect:multipartite}, we derive general bounds for GHZ distribution, which yields an exact expression of GHZ distribution capacity when the most noisy channel is a dephasing channel. Finally, Section \ref{sect:conclusion} discusses the implications of our results and concludes the paper.

\section{Preliminaries}\label{sect:preliminaries}
In this section, we review basic concepts of quantum communication, entanglement generation with classical post-processing, assisted entanglement distillation, and simulation of teleportation covariant channels. Then we proceed to precisely define the primary focus of our paper: entanglement distribution through quantum channels $\mc{N}_1,\mc{N}_2,\cdots,\mc{N}_N$.

\subsection{Quantum Communication and Quantum Capacity}
Given a quantum channel $\mc{N}^{A \rightarrow B}$, an $(n,Q_n,\epsilon_n)$ quantum communication code consists of an encoder $\mc{E}^{S\rightarrow A^n}$ where $\frac{1}{n}\log\dim S = Q_n$ and a decoder $\mc{D}^{B^n\rightarrow S}$ such that 
\begin{align}
    \left\|\psi^{RS} - \mathrm{id}^R\otimes\left(\mc{D}^{B^n\rightarrow S}\circ\mc{N}^{\otimes n}\circ\mc{E}^{S\rightarrow A^n}\left(\psi^{RS}\right)\right) \right\|_1 \leq \epsilon_n
\end{align}
for all pure states $\psi^{RS}$. A rate $Q$ is achievable if there exists a sequence of $(n,Q_n,\epsilon_n)$ quantum communication codes such that $\lim_{n\rightarrow\infty} Q_n = Q$ and $\lim_{n\rightarrow\infty} \epsilon_n = 0$. The quantum capacity of the channel $\mc{N}$, denoted by $Q(\mc{N})$, is the supremum of all achievable rates. It has been established that the quantum capacity of a quantum channel is characterized by its coherent information. The coherent information of a bipartite state $\rho^{AB}$ is given by \cite{wilde-2013}
\begin{align}
    I_c(A\rangle B)_\rho = S(B)-S(AB)\label{eq:coh-info}
\end{align}
where $S(AB)$ and $S(B)$ are the von Neumann entropies of the state $\rho^{AB}$ and the reduced state $\rho^B$ respectively. 
The coherent information of a quantum channel $\mc{N}^{A'\rightarrow B}$ is given by maximizing \eqref{eq:coh-info} over all states of the form $\rho^{AB} = \text{id}^A\otimes\mc{N}^{A'\rightarrow B}(\psi^{AA'})$ with pure input $\psi^{AA'}$:
\begin{align}
    I_c(\mc{N}) = \max_{\psi^{AA'}} I_c(A\rangle B)_{\rho}
\end{align}
Then, the quantum capacity of a quantum channel $\mc{N}$ is given by \cite{lloyd-1997,shor-2002,devetak-2005-private}:
\begin{align}
    Q(\mc{N}) = \lim_{n\rightarrow\infty} \frac{1}{n} I_c(\mc{N}^{\otimes n}) \geq I_c(\mc{N}).
\end{align}

\subsection{Entanglement Generation Assisted by Classical Post-Processing (CPP)}
Given a quantum channel $\mc{N}^{A\rightarrow B}$, an $(n,Q_n,\epsilon_n)$ entanglement generation code assisted by CPP consists of an encoding operation $\mc{E}^{A_0\rightarrow\Tilde{A}A^n}$ where $A_0\cong\mathbb{C}$ is trivial and an LOCC post-processing $\mc{L}^{\tilde{A}B^n\rightarrow\hat{A}\hat{B}}$ such that 
\begin{align}
    \left\|\phi^{\hat{A}\hat{B}} - \mc{L}^{\tilde{A}B^n\rightarrow\hat{A}\hat{B}} \circ \mc{N}^{\otimes n} \circ \mc{E}^{A_0\rightarrow\tilde{A}A^n} \right\|_1 \leq \epsilon_n
\end{align}
where $\phi$ is the maximally entangled state with Schimdt rank $2^{nQ_n}$. A rate $Q$ is achievable assisted by CPP if there exists a sequence of $(n,Q_n,\epsilon_n)$ CPP-assisted entanglement generation codes such that $\lim_{n\rightarrow\infty} Q_n = Q$ and $\lim_{n\rightarrow\infty} \epsilon_n = 0$. The CPP assisted entanglement generation capacity of the channel $\mc{N}$, denoted by $E_{cpp}(\mc{N})$, is the supremum of all achievable entanglement generation rates assisted by CPP. A lower bound of $E_{cpp}(\mc{N})$ is the reverse coherent information \cite{garcia-patron-2009}, which is defined by
\begin{align}
    I_R(\mc{N}) = \max_{\psi^{AA'}} I_c(B\rangle A)_\rho
\end{align}
where again $\rho^{AB} = \text{id}^A\otimes\mc{N}^{A'\rightarrow B}(\psi^{AA'})$. Conversely, an upper bound is given by the Rains information of the channel. The Rains relative entropy of a bipartite state $\rho^{AB}$ is given by \cite{rains-2001}
\begin{align}
    R(A;B)_\rho = \inf_{\sigma\in\mathrm{PPT}'} D(\rho\|\sigma)\label{eq:rains-rel-ent}
\end{align}
where $D(\rho\|\sigma) = \tr\rho(\log\rho-\log\sigma)$ is the quantum relative entropy and $\mathrm{PPT'} = \{\sigma^{AB}:\sigma\geq0, \|\sigma^{T_B}\|_1\leq1\}$ with $T_B$ denoting the partial transpose on system $B$. We can similarly define the Rains information of a quantum channel by optimizing \eqref{eq:rains-rel-ent} over states of the form $\rho^{AB} = \text{id}^A\otimes\mc{N}^{A'\rightarrow B}(\psi^{AA'})$ with $\psi^{AA'}$ pure \cite{tomamichel-2017}:
\begin{align}
    R(\mc{N}) = \max_{\psi^{AA'}} R(A;B)_\rho.
\end{align}

\subsection{Assisted Entanglement Distillation}
Suppose $n$ copies of a tripartite mixed state $\rho^{ABC}$ are shared between $A$, $B$, and a helper $C$, and the goal is to distill bipartite entanglement between $A$ and $B$ with the help of $C$, who assists by performing a measurement on his system and broadcasting the result to $A$ and $B$ (Fig.~\ref{fig:ass-ent-dist}). More precisely, a $(n,Q_n,\epsilon_n)$-assisted entanglement distillation protocol with $C$ being the helper consists of a POVM $\{M_i^{C^n}\}_i$ and the corresponding LOCC operations $\mc{L}_i^{A^nB^n\rightarrow\hat{A}\hat{B}}$ such that
\begin{align}
    \left\|\phi^{\hat{A}\hat{B}} - \sum_i \mc{L}_i^{A^nB^n\rightarrow\hat{A}\hat{B}}\left(\tr_{C^n}\left[M_i^{C^n}\left(\rho^{ABC}\right)^{\otimes n}\right]\right)\right\|_1 \leq \epsilon_n
\end{align}
where $\phi$ is the maximally entangled state with Schimdt rank $2^{nQ_n}$. A rate $Q_n$ is achievable if there exists a sequence of $(n,Q_n,\epsilon_n)$-assisted entanglement distillation protocols such that $\lim_{n\rightarrow\infty} Q_n = Q$ and $\lim_{n\rightarrow\infty} \epsilon_n = 0$. Dutil and Hayden showed that the supremum of all achievable rates is given by \cite{dutil2011assisted}
\begin{align}
    D_A^\infty(\rho^{ABC}) = \lim_{n\rightarrow\infty} \frac{1}{n} D_A\left(\left(\rho^{ABC}\right)^{\otimes n}\right)
\end{align}
where
\begin{align}
    D_A(\rho^{ABC}) = \sup_{\{M_i\}_i} \left\{\sum_i p_iD(\sigma^{AB}_i): p_i=\tr(M_i^C\rho^{ABC}),\ \sigma_i = \frac{1}{p_i}\tr_C(M_i^C\rho^{ABC})\right\},
\end{align}
where the supremum is taken over all POVMs $\{M_i\}_i$ on $C$. We will call $D_A(\rho)$ the entanglement of assistance and $D_A^\infty(\rho)$ the asymptotic entanglement of assistance.
\begin{figure}
    \centering
    \subfloat[]{\includegraphics[width=0.49\linewidth]{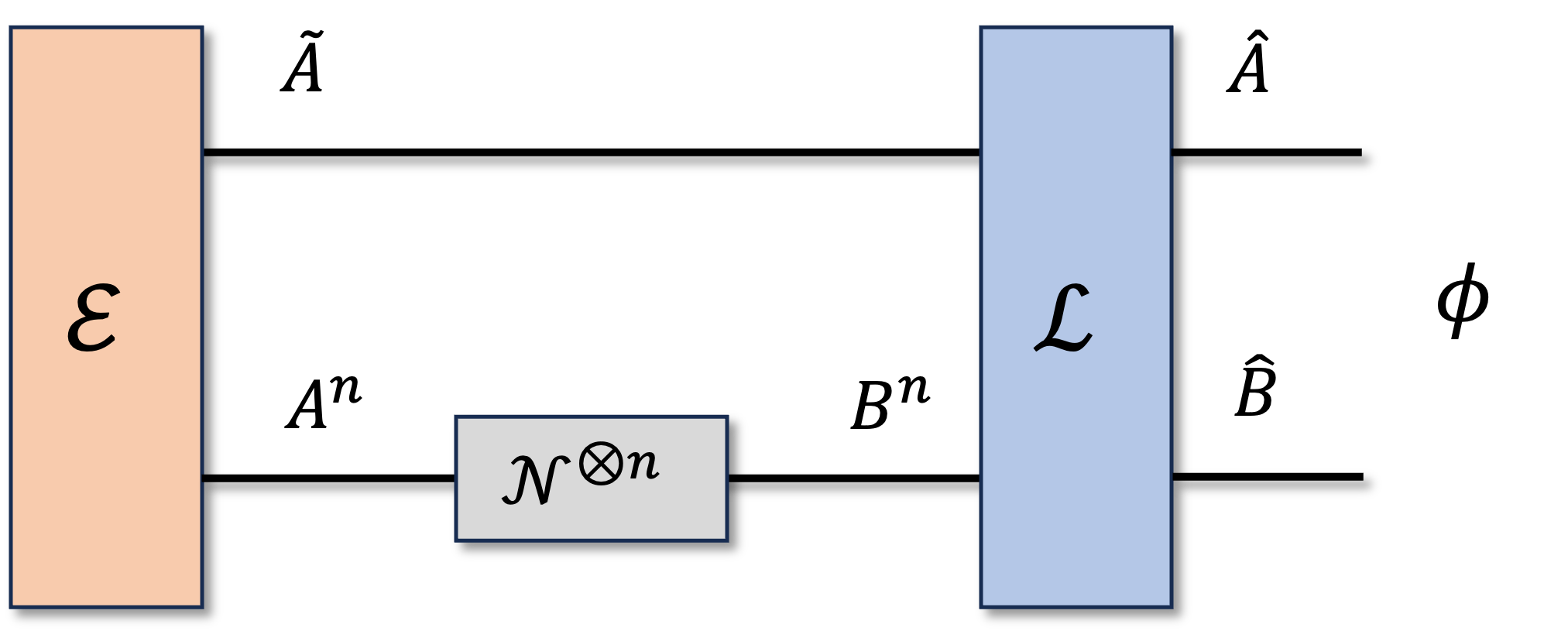}\label{fig:E_cpp}}
    \hfill
    \subfloat[]{\includegraphics[width=0.5\linewidth]{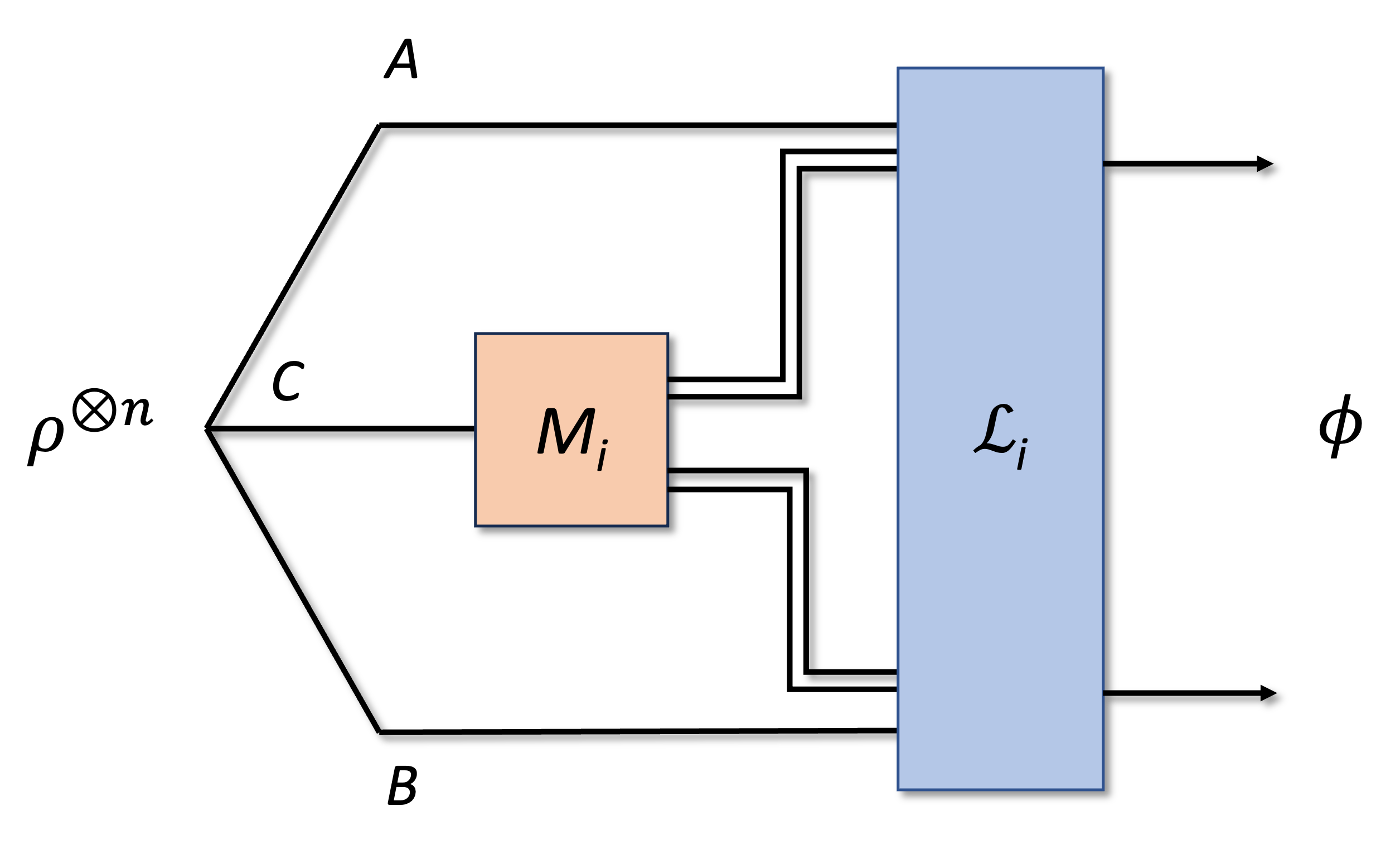}\label{fig:ass-ent-dist}}
    \caption{(a) A diagram for entanglement generation with CPP assistance. (b) A diagram for assisted entanglement distillation.}
    
\end{figure}

\subsection{Simulating Teleportation-Covariant Channels}
Consider a qudit quantum teleportation protocol \cite{bennett-1993}. A $d$-dimensional maximally entangled state $\phi_d^{AB}$ is shared between two parties $A$ and $B$. In addition, $A$ possesses an arbitrary quantum state $\rho^{A'}$. Suppose $A$ performs a measurement on $AA'$ in the basis $\left\{\mbb{I}^A\otimes\left(X_d^aZ_d^b\right)^{\dagger A'}\ket{\phi_d}^{AA'}\right\}_{a,b}$. Here $X_d$ is the shift unitary defined by $X_d\ket{k} = \ket{k\oplus1}$ where $\oplus$ denotes addition modulo $d$, and the clock unitary $Z_d$ is defined by $Z_d\ket{k} = e^{i2\pi k/d}\ket{k}$. The set of unitaries of the form $U_{ab} = X_d^aZ_d^b$ are called generalized Pauli unitaries. After the measurement, the state on $B$ becomes $X_d^aZ_d^b\rho(X_d^aZ_d^b)^\dagger$, depending on the measurement result $a,b$. Once $B$ learns the measurement result from $A$, the state can be corrected back to $\rho$ by performing the appropriate unitary. 

A qudit quantum channel $\mc{N}^{A\rightarrow B}$ is called teleportation-covariant if for any generalized Pauli unitary $U^{A}$ and for any input state $\rho^{A}$, we have 
\begin{align}
    \mc{N}^{A\rightarrow B}(U^A\rho^A U^{\dagger A}) = V^B\mc{N}^{A\rightarrow B}(\rho^A)V^{\dagger B}
\end{align}
for some unitary $V^{B}$. Teleportation-covariant channels can be simulated by their Choi state $J_{\mc{N}}^{AB}=\mathrm{id}^A\otimes\mc{N}^{A'\rightarrow B}(\phi_d^{AA'})$ and one-way LOCC \cite{pirandola2017fundamental}. To see this, note that if we replace $\phi_d$ with $J_{\mc{N}}$ in the teleportation protocol, the state on $B$ after measurement on $AA'$ becomes $\mc{N}(U_{ab}\rho U_{ab}^\dagger)$. By teleportation covariance this is equivalent to $V_{ab}\mc{N}(\rho)V_{ab}^\dagger$ for some unitary $V_{ab}$ depending on the measurement result. Once $B$ learns the measurement result, one can apply an appropriate unitary to correct the state back to $\mc{N}(\rho)$.

\subsection{Entanglement Distribution from a Central Source}\label{sect:entanglement-distribution-definition}
In this section, we rigorously define EPR and GHZ distribution capacities, which are the focal points in our work. 
Given $N$ quantum channels $\{\mc{N}_i^{S_i\rightarrow A_i}\}_{i=1}^N$, an $(n,Q_n,\epsilon_n)$-$\tau$ distribution code consists of an encoding operation $\mc{E}^{S_0\rightarrow S_1^n\cdots S_N^{\otimes n}}$
where $S_0\cong\mathbb{C}$ is trivial and an LOCC post-processing $\mc{L}^{A_1^{\otimes n}\cdots A_N^{\otimes n}\rightarrow\widehat{A_1}\cdots\widehat{A_N}}$ such that
\begin{align}
    \left\| \left(\tau^{\otimes k}\right)^{\widehat{A_1}\cdots\widehat{A_N}} - \mc{L} \circ \left(\mc{N}_1^{\otimes n}\otimes\cdots\otimes\mc{N}_N^{\otimes n}\right) \circ \mc{E}\right\|_1 \leq \epsilon_n
\end{align}
where $Q_n=k/n$. A pictorial description can be found in Fig.~\ref{fig:multi-sided-diagram}. We will focus on the case where $\tau=\state{\mathrm{GHZ}_N}$. A rate $Q$ is achievable if there exists an $(n,Q_n,\epsilon_n)$-GHZ distribution code such that $\lim_{n\rightarrow\infty} Q_n = Q$ and $\lim_{n\rightarrow\infty} \epsilon_n = 0$. The GHZ distribution capacity $E(\mc{N}_1,\cdots,\mc{N}_N)$ is the supremum of all achievable rates. When $N=2$, this similarly defines the EPR distribution capacity $E(\mc{N}_1,\mc{N}_2)$.

\section{Connection between Entanglement Distribution and Assisted Entanglement Distillation}\label{sect:duality}
We begin with the simplest case where $N=2$. Our first result is a connection between EPR distribution capacity $E(\mc{N}_1,\mc{N}_2)$ and the assisted entanglement distillation rate on the Choi states $D_A^\infty\left(J_{\mc{N}_1}^{A_1S_1}\otimes J_{\mc{N}_2}^{A_2S_2}\right)$, where $S$ acts as the assistant party. 

\begin{theorem}
    $E(\mc{N}_1,\mc{N}_2) \geq D_A^\infty\left(J_{\mc{N}_1}^{A_1S_1}\otimes J_{\mc{N}_2}^{A_2S_2}\right)$. Moreover, if both $\mc{N}_1$ and $\mc{N}_2$ are teleportation-covariant, then $E(\mc{N}_1,\mc{N}_2) = D_A^\infty\left(J_{\mc{N}_1}^{A_1S_1}\otimes J_{\mc{N}_2}^{A_2S_2}\right)$.
\end{theorem}
\begin{proof}
    For an arbitrary $n$, a POVM $\{M_i\}_i$ on $S_1^nS_2^n$ yields the ensemble
    \begin{align}
        p_i &= \tr\left\{M_i\left[\left(J_{\mc{N}_1}^{S_1A_1}\right)^{\otimes n}\otimes \left(J_{\mc{N}_2}^{S_2A_2}\right)^{\otimes n}\right]\right\} = \frac{\tr M_i}{d_{S_1}^nd_{S_2}^n}\\
        \sigma_i &= \frac{1}{p_i}\tr_{S_1^nS_2^n}\left\{M_i\left[\left(J_{\mc{N}_1}^{S_1A_1}\right)^{\otimes n}\otimes \left(J_{\mc{N}_2}^{S_2A_2}\right)^{\otimes n}\right]\right\} = \mc{N}_1^{\otimes n}\otimes\mc{N}_2^{\otimes n}\left(\frac{M_i^T}{\tr M_i^T}\right)
    \end{align}
    where $d_{S_i}=\dim S_i$, and we have invoked the transpose property $\mbb{I}^A\otimes M^B\ket{\phi}^{AB}=(M^T)^A\otimes\mbb{I}^B\ket{\phi}^{AB}$. This means that performing a POVM $\{M_i\}$ on $J_{\mc{N}_1}^{\otimes n}\otimes J_{\mc{N}_2}^{\otimes n}$ followed by the corresponding LOCC post-processing $\mc{L}_i$ is equivalent to probabilistically sending states $M_i^T/\tr M_i^T$ through $\mc{N}_1^{\otimes n}\otimes\mc{N}_2^{\otimes n}$ followed by the same post-processing $\mc{L}_i$. The latter is clearly less desirable than deterministically performing the optimal EPR distribution protocol. 

    For the second part of the theorem, recall that if a channel is teleportation-covariant, its action can be simulated by a teleportation protocol with its Choi state as the resource state. Therefore, an EPR distribution protocol with input state $\rho^{S_1^nS_2^n}$ and LOCC post-processing is equivalent to the scheme illustrated in Fig.~\ref{fig:teleportation-covariance}. Since any state that is local to $S$ cannot help distill bipartite entanglement between $A_1$ and $A_2$, this is further equivalent to assisted entanglement distillation on $(J_{\mc{N}_1}\otimes J_{\mc{N}_2})^{\otimes n}$. It then follows that $E(\mc{N}_1,\mc{N}_2) \leq D_A^\infty\left(J_{\mc{N}_1}^{A_1S_1}\otimes J_{\mc{N}_2}^{A_2S_2}\right)$.
\end{proof}
\begin{figure}
    \centering
    \includegraphics[width=0.6\textwidth]{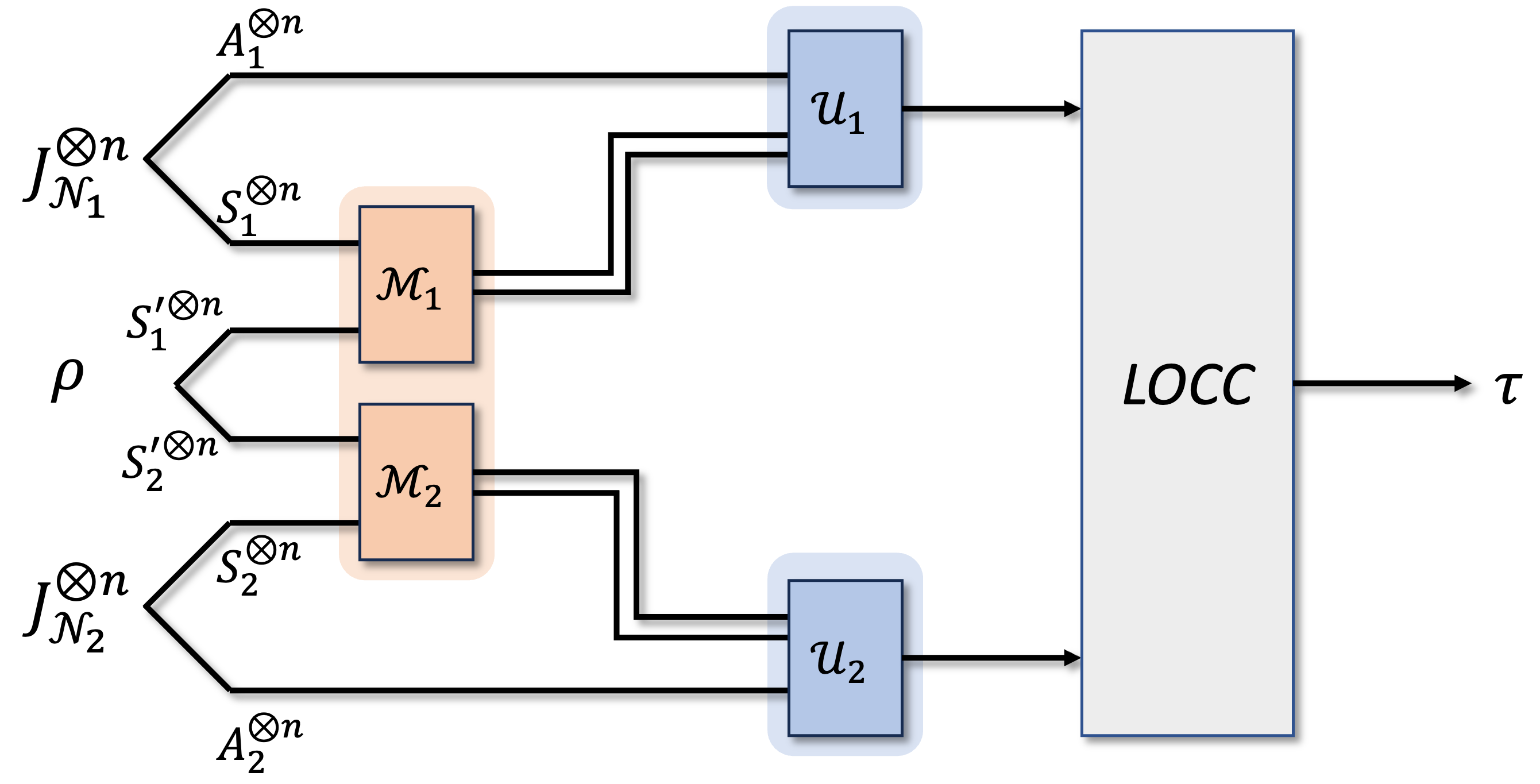}
    \caption{Simulating teleportation-covariant channels with their Choi states.}
    \label{fig:teleportation-covariance}
\end{figure}

In general, the asymptotic entanglement of assistance $D_A^\infty(\rho)$ is very hard to mathematically characterize. To simplify its characterization, we introduce the Rains relative entropy of assistance $R_A(\rho)$ and its regularization, namely the asymptotic Rains relative entropy of assistance $R_A^\infty(\rho)$:
\begin{align}
    R_A^\infty(\rho^{ABC}) &\coloneqq \lim_{n\rightarrow\infty} \frac{1}{n} R_A\left(\left(\rho^{ABC}\right)^{\otimes n}\right) \\
    R_A(\rho^{ABC}) &\coloneqq \sup_{\{E_i\}} \left\{\sum_i p_i R(A;B)_{\sigma_i}: p_i=\tr(E_i^C\rho^{ABC})\ \sigma_i = \frac{1}{p_i}\tr_C(E_i^C\rho^{ABC})\right\}, 
\end{align}
where $R(A;B)_\rho$ is the Rains relative entropy introduced in Section \ref{sect:preliminaries}. By definition $R_A(\rho)\geq D_A(\rho)$ and $R_A^\infty(\rho)\geq D_A^\infty(\rho)$. Desirably, the Rains relative entropy of assistance $R_A(\rho)$ is convex, and we will later apply this property to characterize the EPR distribution capacity of two erasure channels.
\begin{proposition}
    The Rains information of assistance is convex.
\end{proposition}
\begin{proof}
    Consider an arbitrary decomposition $\rho^{ABC}=\sum_j q_j\rho_j^{ABC}$ and an arbitrary POVM $\{M_i\}_i$ on $C$. Supposing that $p_i=\tr(M_i^C\rho^{ABC})$, the measurement yields states
    \begin{align}
        \sigma_i^{AB} = \frac{1}{p_i} \tr_C(M_i^C\rho^{ABC}) &= \frac{1}{p_i} \sum_j q_j \tr_C(M_i^C\rho_j^{ABC}) = \sum_j \frac{q_jp_{i|j}}{p_i} \sigma_{ij}^{AB}
    \end{align}
    where $p_{i|j}=\tr(M_i^C\rho_j^{ABC})$ and $\sigma_{ij}^{AB} = \tr_C(M_i^C\rho_j^{ABC})/\tr(M_i^C\rho_j^{ABC})$. Therefore,
    \begin{align}
        \sum_i p_i R(A;B)_{\sigma_i} \leq \sum_i p_i \sum_j \frac{q_jp_{i|j}}{p_i} R(A;B)_{\sigma_{ij}} = \sum_{i,j} q_j p_{i|j} R(A;B)_{\sigma_{ij}} 
    \end{align}
    where the inequality follows from the convexity of Rains information. Optimizing over all POVMs on both sides of the inequality we obtain 
    \begin{align}
        R_A(\rho^{ABC}) \leq \sum_j q_j R_A(\rho_j^{ABC}),
    \end{align}
    which concludes the proof.
\end{proof}


\subsection{Application to Erasure channels}\label{sect:erasure-channels}
Consider qudit erasure channels of the form $\mc{E}_p(\rho) = (1-p)\rho + p\state{e}$, where $\ket{e}$ is a flag state orthogonal to the input space, signaling the erasure of the input. As an application of the previous results, we exactly characterize the EPR distribution capacity of two qudit erasure channels.
\begin{theorem}
    $E(\mc{E}_{p_1},\mc{E}_{p_2}) = (1-p_1)(1-p_2)\log d$.
\end{theorem}
\begin{proof}
    Clearly, $(1-p_1)(1-p_2)\log d$ is achievable by sending $\phi_d$ through the channels. With probability $(1-p_1)(1-p_2)$ no erasure occurs, whereas with probability $1-(1-p_1)(1-p_2)$ an erasure error occurs on at least one of $A_1$ and $A_2$. They can locally detect erasure errors and classically communicate if an erasure error occurred and discard any such states. This results in a rate of $(1-p_1)(1-p_2)\log d$. Conversely, since erasure channels are teleportation-covariant, $E(\mc{E}_{p_1},\mc{E}_{p_2}) = D_A^\infty\left(J_{\mc{E}_{p_1}}^{S_1A_1} \otimes J_{\mc{E}_{p_2}}^{S_2A_2}\right)\leq R_A^\infty\left(J_{\mc{E}_{p_1}}^{S_1A_1} \otimes J_{\mc{E}_{p_2}}^{S_2A_2}\right)$. Now, for any $n$,
    \begin{align}
        &\left(J_{\mc{E}_{p_1}}^{S_1A_1}\right)^{\otimes n} \otimes \left(J_{\mc{E}_{p_2}}^{S_2A_2}\right)^{\otimes n} \notag\\
        =& \left((1-p_1)\phi_d^{S_1A_1}+p_1\frac{\mbb{I}}{d}^{S_1}\otimes\state{e}^{A_1}\right)^{\otimes n} \otimes \left(J_{\mc{E}_{p_2}}^{S_2A_2}\right)^{\otimes n}\notag\\
        =& \sum_k \binom{n}{k}(1-p_1)^{k}p_1^{n-k}        \left(\phi_d^{S_1A_1}\right)^{\otimes k}\otimes\left(\frac{\mbb{I}}{d}^{S_1}\otimes\state{e}^{A_1}\right)^{\otimes n-k} \otimes \left(J_{\mc{E}_{p_2}}^{S_2A_2}\right)^{\otimes n} \; .
    \end{align}
    Focusing on each term, the Rains relative entropy of assistance is
    \begin{align}
        &R_A\left(\left(\phi_d^{S_1A_1}\right)^{\otimes k}\otimes\left(\frac{\mbb{I}}{d}^{S_1}\otimes\state{e}^{A_1}\right)^{\otimes n-k} \otimes \left(J_{\mc{E}_{p_2}}^{S_2A_2}\right)^{\otimes n}\right) \notag\\
        =& R_A\left(\left(\phi_d^{S_1A_1}\right)^{\otimes k}\otimes\left(J_{\mc{E}_{p_2}}^{S_2A_2}\right)^{\otimes n}\right) \leq k(1-p_2)\log d.
    \end{align}
    The equality follows from the fact that states local to $S$ and $A_1$ cannot help. The proof of the inequality is given in the lemma below. Therefore, by convexity of $R_A$,
    \begin{align}
        &R_A\left(\left(J_{\mc{E}_{p_1}}^{S_1A_1}\right)^{\otimes n} \otimes \left(J_{\mc{E}_{p_2}}^{S_2A_2}\right)^{\otimes n}\right) \\
        &\leq \sum_i \binom{n}{i} (1-p_1)^{i}p_1^{n-i} R_A\left(\left(\phi_d^{S_1A_1}\right)^{\otimes i}\otimes\left(\frac{\mbb{I}}{d}^{S_1}\otimes\state{e}^{A_1}\right)^{\otimes n-i} \otimes \left(J_{\mc{E}_{p_2}}^{S_2A_2}\right)^{\otimes n}\right) \\
        &\leq \sum_i \binom{n}{i} (1-p_1)^{i}p_1^{n-i}i (1-p_2)\log d = (1-p_1)(1-p_2)\log d.
    \end{align}
\end{proof}
\begin{lemma}
    $R_A\left(\left(\phi_d^{S_1A_1}\right)^{\otimes k}\otimes\left(J_{\mc{E}_{p_2}}^{S_2A_2}\right)^{\otimes n}\right) \leq k(1-p_2)\log d$.
\end{lemma}
\begin{proof}
    After performing a POVM $\{M_i\}_i$ on $S_1^kS_2^n$ system of the given state, the post-measurement state conditioned on measurement outcome $i$ is
    \begin{align}
        \sigma_i = \left[\text{id}^{S_1^k}\otimes\left(\mc{E}_{p_2}^{S_2\rightarrow A_2}\right)^{\otimes n}\right](\rho_i^{S_1^kS_2^n}), \qquad \rho_i=\frac{M_i^T}{\tr M_i^T} \; .
    \end{align}
    Any state $\rho^{S_1^kS_2^n}$ can be simulated by a maximally entangled state $\phi^{S_1^kS_2^k}$ and one-way LOCC, i.e., $\rho^{S_1^kS_2^n}=\sum_j \mc{M}_j^{S_1^k}\otimes\mc{O}_j^{S_2^n}(\phi^{S_1^kS_2^k}\otimes\tau^{S_2^{n-k}})$, where $\tau$ is an arbitrary state, $\mc{M}_j$ are completely positive trace-non-increasing maps such that $\sum_j\mc{M}_j$ is trace-preserving, and $\mc{O}_j$ are completely positive and trace-preserving (CPTP). For any CPTP map $\mc{O}^{S_2^n}$, we have $\mc{E}^{\otimes n}\circ\mc{O}=\tilde{\mc{O}}\circ\mc{E}^{\otimes n}$, where $\tilde{\mc{O}}$ is the CPTP map that measures the location of erasure errors and implements the partial trace of $\mc{O}$ on the intact systems. In mathematical terms, $\tilde{\mc{O}}(\rho)=\sum_{L\subset S_2^n} \mc{O}_{L^c}(P_L\rho P_L)$, where the sum is over subsystems $L$ of $S_2^n$, operators $P_L$ are projections associated with measuring error flag on and only on subsystems $L$, and $\mc{O}_{L^c}=\tr_{L}\mc{O}$. Therefore, 
    \begin{align}
        \sigma_i = \sum_j \left(\mc{M}_{i,j}^{S_1^k}\otimes\tilde{\mc{O}}_{i,j}^{A_2^n}\right) \circ \left[\text{id}^{S_1^k}\otimes\left(\mc{E}_{p_2}^{S_2\rightarrow A_2}\right)^{\otimes n}\right]\left(\phi^{S_1^kS_2^k}\otimes\tau^{S_2^{n-k}}\right) \; .
    \end{align}
    By monotonicity of Rains relative entropy, 
    \begin{align}
        R(\sigma_i) &\leq R\left(\left[\text{id}^{S_1^k}\otimes\left(\mc{E}_{p_2}^{S_2\rightarrow A_2}\right)^{\otimes n}\right]\left(\phi^{S_1^kS_2^k}\otimes\tau^{S_2^{n-k}}\right)\right) \\
        &= R\left(\left[\text{id}^{S_1}\otimes\mc{E}_{p_2}^{S_2\rightarrow A_2}(\phi^{S_1S_2})\right]^{\otimes k} \otimes \left(\mc{E}_{p_2}^{S_2\rightarrow A_2}\right)^{\otimes n-k}(\tau^{S_2^{n-k}})\right).
    \end{align}
    Since $\left(\mc{E}_{p_2}^{S_2\rightarrow A_2}\right)^{\otimes n-k}(\tau^{S_2^{n-k}})$ is a state local to $A_2$, it does not increase the Rains relative entropy. Thus,
    \begin{align}
        R(\sigma_i) &\leq R\left(\left[\text{id}^{S_1}\otimes\mc{E}_{p_2}^{S_2\rightarrow A_2}(\phi^{S_1S_2})\right]^{\otimes k}\right) \\ 
        &= k(1-p_2)\log d \;,
    \end{align}
    which concludes the proof.
\end{proof}

\section{An Explicit Protocol Construction for EPR Distribution}\label{sect:explicit-protocol}
While we have established a connection between $E(\mc{N}_1,\mc{N}_2)$ and $D_A^\infty(J_{\mc{N}_1}\otimes J_{\mc{N}_2})$ and used the connection to give an exact expression of the EPR distribution capacity of two erasure channels, it is in general difficult to mathematically characterize the asymptotic entanglement of assistance $D_A^\infty$. In this section, we introduce a simple lower bound of the EPR distribution capacity based on a direct construction.
\begin{theorem}\label{thm:construction}
    $E(\mc{N}_1,\mc{N}_2) \geq \max\{\min\{Q(\mc{N}_1),E_{cpp}(\mc{N}_2)\},\min\{Q(\mc{N}_2),E_{cpp}(\mc{N}_1)\}\}$.
\end{theorem}
\begin{proof}
    Let $Q\coloneqq\min\{Q(\mc{N}_1),E_{cpp}(\mc{N}_2)\}$. By the definitions of unassisted quantum capacity and CPP assisted entanglement generation capacity, for any $\epsilon_1,\epsilon_2\in(0,1)$ and $\delta>0$, there exists an $(n,Q-\delta,\epsilon_1)$ quantum communication code $(\mc{E}_1,\mc{D}_1)$ for $\mc{N}_1$ and an $(n,Q-\delta,\epsilon_2)$ CPP assisted entanglement generation code $(\mc{E}_2,\mc{L}_2)$ for $\mc{N}_2$, as long as $n$ is sufficiently large. Our GHZ distribution protocol then consists of an encoding operation $\mc{E} \coloneqq \mc{E}_1^{S\rightarrow S_1^n} \circ \mc{E}_2^{S_0\rightarrow SS_2^n}$ and classical post-processing $\mc{L} \coloneqq \mc{L}_2^{SA_2^{\otimes n}\rightarrow\hat{A}\hat{B}} \circ \mc{D}_1^{A_1^{\otimes n}\rightarrow S}$. The error of the protocol is 
    \begin{align}
        \left\| \phi_d - \mc{L}\circ\left(\mc{N}_1^{\otimes n}\otimes\mc{N}_2^{\otimes n}\right)\circ\mc{E} \right\|_1 &\leq \left\| \phi_d - \mc{L}_2\circ\mc{N}_2^{\otimes n}\circ\mc{E}_2 \right\|_1 + \left\| \mc{L}_2\circ\mc{N}_2^{\otimes n}\circ\mc{E}_2 - \mc{L}\circ\left(\mc{N}_1^{\otimes n}\otimes\mc{N}_2^{\otimes n}\right)\circ\mc{E} \right\|_1 \\
        &\leq \epsilon_2 + \left\| \mc{E}_2 - \mc{D}\circ\mc{N}_1^{\otimes n}\circ\mc{E}_1\circ\mc{E}_2 \right\|_1 \\
        &\leq \epsilon_1 + \epsilon_2.
    \end{align}
    where the first inequality follows from the triangle inequality, the second inequality follows from the fact that $(\mc{E}_2,\mc{L}_2)$ is a CPP assisted entanglement generation code and the data processing inequality, and the final inequality follows from the fact that $(\mc{E}_1,\mc{D}_1)$ is a quantum communication code over $\mc{N}_1$. Therefore, we conclude that for any $\epsilon_1,\epsilon_2\in(0,1)$ and $\delta>0$, there exists a $(n,Q-\delta,\epsilon_1+\epsilon_2)$ EPR distribution protocol for sufficiently large $n$. This means that $Q$ is an achievable EPR distribution rate. Interchanging $\mc{N}_1$ and $\mc{N}_2$ proves we can achieve $\min\{Q(\mc{N}_2),E_{cpp}(\mc{N}_1)\}$ as well, which establishes the theorem.
\end{proof}
\begin{remark}
    Since $E_{cpp}(\mc{N})\geq Q(\mc{N})$ for an arbitrary channel $\mc{N}$, we deduce $E(\mc{N}_1,\mc{N}_2)\geq\min\{Q(\mc{N}_1),Q(\mc{N}_2)\}$.
\end{remark}
\begin{remark}
    Evidently, $E(\mc{N}_1,\mc{N}_2)\leq E(\mc{N}_1,\mathrm{id})\leq E_{cpp}(\mc{N}_1)$. By the same argument $E(\mc{N}_1,\mc{N}_2)\leq E_{cpp}(\mc{N}_2)$ as well. Thus, $E(\mc{N}_1,\mc{N}_2)\leq\min\{E_{cpp}(\mc{N}_1),E_{cpp}(\mc{N}_2)\}\leq\min\{R(\mc{N}_1),R(\mc{N}_2)\}$. A consequence of this observation is that the lower bound in Thm. \ref{thm:construction} is tight when the quantum capacity of a channel is greater than the entanglement generation capacity of the other channel. To see this, suppose without loss of generality that $Q(\mc{N}_1)\geq E_{cpp}(\mc{N}_2)$, then $E_{cpp}(\mc{N}_1)\geq Q(\mc{N}_1)\geq E_{cpp}(\mc{N}_2)\geq Q(\mc{N}_2)$. The lower bound thus implies $E(\mc{N}_1,\mc{N}_2)\geq \max\{E_{cpp}(\mc{N}_2),Q(\mc{N}_2)\}=E_{cpp}(\mc{N}_2)$.
\end{remark}

\subsection{Generalized Amplitude Damping Channels}\label{sect:GADC}
Generalized amplitude damping channels $\mc{A}_{\gamma,T}$ model the $T_1$ relaxation process due to the coupling of a spin-1/2 system to a thermal reservoir at a higher temperature than that of the spin system \cite[p. 382]{nielsen-2010}. The Kraus operators 
\begin{align}
    K_1 &= \sqrt{1-T}\left(\state{0}+\sqrt{1-\gamma}\state{1}\right), \\
    K_2 &= \sqrt{\gamma(1-T)}\op{0}{1}, \\
    K_3 &= \sqrt{T}\left(\sqrt{1-\gamma}\state{0}+\state{1}\right), \\
    K_4 &= \sqrt{\gamma T}\op{1}{0}.
\end{align}
If the thermal parameter $T=0$, then $K_3$ and $K_4$ vanish, and the channel reduces to the ordinary amplitude damping channel $\mc{A}_{\gamma}\coloneqq\mc{A}_{\gamma,0}$, which models the case where the thermal reservoir is at zero temperature. While there are no exact expressions for the quantum capacity or the CPP-assisted entanglement generation capacity for GADCs presently, the coherent information and the Rains information of GADCs have been computed in Refs. \cite{garcia-patron-2009} and Ref. \cite{khatri-2020} respectively. Note that a GADC $\mc{A}_{\gamma,T}$ is anti-degradable when $\gamma\geq1/2$ and thus its quantum capacity necessarily vanishes, while the Rains information does not necessarily vanish. This means there is potentially a large gap between the lower bound given above and the upper bound. This motivates us to find better lower bounds. 

In the following, we detail an improved protocol based on multi-rail encoding \cite{chuang1996quantum,duan2010multi,abdelhadi2024adaptive}. Let us illustrate the idea with amplitude damping channels and dual-rail encoding. Using dual-rail encoding, the receivers are able to detect damping errors. More precisely, we send the state $\ket{\psi_2} = \frac{1}{\sqrt{2}}(\ket{01}\ket{01}+\ket{10}\ket{10})$ through $\mc{A}_{\gamma_1}^{\otimes 2}\otimes\mc{A}_{\gamma_2}^{\otimes 2}$: 
\begin{align}
    \rho = \mc{A}_{\gamma_1}^{\otimes 2}\otimes\mc{A}_{\gamma_2}^{\otimes 2}(\psi_2) &= (1-\gamma_1)(1-\gamma_2)\psi_2 + \frac{\gamma_1(1-\gamma_2)}{2}(\state{0001}+\state{0010}) \\& \qquad {}+\frac{(1-\gamma_1)\gamma_2}{2}(\state{0100}+\state{1000})+\gamma_1\gamma_2\state{0000}
\end{align}
$A_1$ and $A_2$ then both post-select on the one-particle sector, i.e., they both perform the projection $P=\state{01}+\state{10}$, and the post-selected state is $\rho'=(P\otimes P)\rho(P\otimes P)/\tr(P\otimes P)\rho(P\otimes P)$. This procedure produces a rank-2 maximally entangled state with probability $(1-\gamma_1)(1-\gamma_2)$, meaning we can achieve the rate of $(1-\gamma_1)(1-\gamma_2)/2$ per channel use. More generally, we can use multi-rail encoding:
\begin{align}
    \ket{\psi_k} = \frac{1}{\sqrt{k}}(\ket{00\cdots01}\ket{00\cdots01}+\ket{00\cdots10}\ket{00\cdots10}+\cdots+\ket{10\cdots00}\ket{10\cdots00}).
\end{align}
The receivers again post-select on the one-particle sector, producing a rank-$k$ maximally entangled state with probability $(1-\gamma_1)(1-\gamma_2)$, thus achieving the rate of $(1-\gamma_1)(1-\gamma_2)\log k/k$. It is easy to check that $\sup_{k\in\mbb{Z}_+} \log k/k$ is achieved when $k=3$. Therefore, the largest rate that we can achieve with multi-rail encoding is $(1-\gamma_1)(1-\gamma_2)\log 3/3$.

This strategy can be adapted to generalized amplitude damping channels. We again use multi-rail encoding and post-select on the one-particle subspace. After post-selection, the two receivers perform the hashing protocol \cite{devetak-2005-distillation}. The rates achievable with dual-rail encoding are plotted in Fig. \ref{fig:GADC}. 

\begin{figure}
    \centering
    \includegraphics[width=\textwidth]{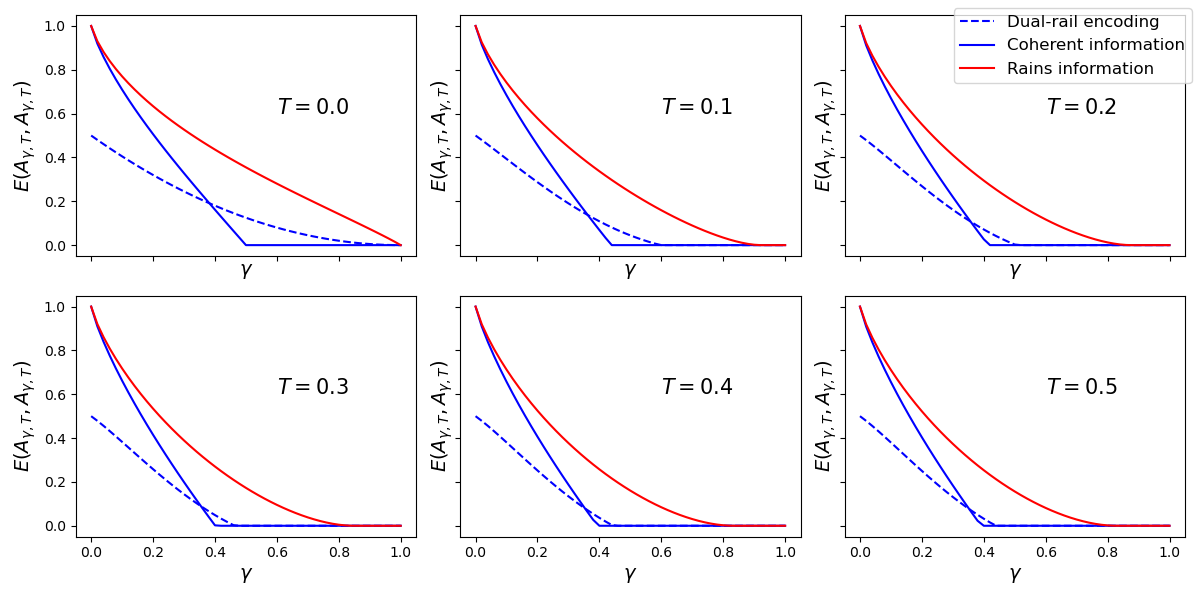}
    \caption{Lower and upper bounds of the bipartite entanglement distribution capacity, assuming the two channels are generalized amplitude damping channels with the same parameters.}
    \label{fig:GADC}
\end{figure}

\section{GHZ distribution}\label{sect:multipartite}
In this section, we establish bounds for the GHZ distribution capacity $E(\mc{N}_1,\cdots,\mc{N}_N)$ with $N>2$. 
\begin{theorem}
    $E(\mc{N}_1,\cdots,\mc{N}_N)\geq\min_i Q(\mc{N}_i)$.
    \label{thm:general-lower-bound}
\end{theorem}
\begin{proof}
    Let $Q \coloneqq \min_i Q(\mc{N}_i)$. By definition of unassisted quantum capacity, for each channel $\mc{N}_i$ and for all $\epsilon_i\in(0,1), \delta>0$, there exists an $(n,Q-\delta,\epsilon_i)$ quantum communication code $(\mc{E}_i,\mc{D}_i)$ for sufficiently large $n$. Our GHZ distribution protocol then consists of an encoding operation $\mc{E} \coloneqq \bigotimes_i \mc{E}_i^{S_i\rightarrow S_i^{\otimes n}}(\state{\mathrm{GHZ}}^{S_1\cdots S_N})$ and post-processing $\mc{D} \coloneqq \bigotimes_i \mc{D}_i^{A_i^{\otimes n}\rightarrow \hat{A}_i}$. The error of the protocol is
    \begin{align}
        &\left\| \state{\mathrm{GHZ}} - \left(\bigotimes\nolimits_i \mc{D}_i\circ\mc{N}_i\circ\mc{E}_i\right)(\state{\mathrm{GHZ}}) \right\|_1 \\
        & \leq \left\| \state{\mathrm{GHZ}} - \mc{D}_1\circ\mc{N}_1\circ\mc{E}_1(\state{\mathrm{GHZ}} ) \right\|_1 \\
        &\qquad+ \left\|\mc{D}_1\circ\mc{N}_1\circ\mc{E}_1(\state{\mathrm{GHZ}} ) - \left(\bigotimes\nolimits_i \mc{D}_i\circ\mc{N}_i\circ\mc{E}_i\right)(\state{\mathrm{GHZ}}) \right\|_1 \\
        & \leq \epsilon_1 + \left\| \state{\mathrm{GHZ}} - \left(\bigotimes\nolimits_{i\neq1} \mc{D}_i\circ\mc{N}_i\circ\mc{E}_i\right)(\state{\mathrm{GHZ}}) \right\|_1
    \end{align}
    where in the first inequality follows from the triangle inequality, and the second inequality follows from the fact that $(\mc{E}_1,\mc{D}_1)$ is a quantum communication code for $\mc{N}_1$ and the data processing inequality for trace distance. Following this procedure iteratively, we find that 
    \begin{align}
        \left\| \state{\mathrm{GHZ}} - \left(\bigotimes\nolimits_i \mc{D}_i\circ\mc{N}_i\circ\mc{E}_i\right)(\state{\mathrm{GHZ}}) \right\|_1 \leq \sum\nolimits_{i} \epsilon_i
    \end{align}
    Therefore, we conclude that for any $\epsilon_i\in(0,1)$ and $\delta>0$, there exists a $\left(n,Q-\delta,\sum_i\epsilon_i\right)$ GHZ distribution protocol for sufficiently large $n$. In other words, $Q$ is an achievable rate.
\end{proof}
\noindent Conversely, we establish an upper bound in terms of CPP-assisted entanglement generation capacity.
\begin{theorem}
    $E(\mc{N}_1,\cdots,\mc{N}_N) \leq \min_i E_{\mathrm{cpp}}(\mc{N}_i)$.
    \label{thm:general-upper-bound}
\end{theorem}
\begin{proof}
    For an arbitrary receiver $A_i$, we ignore the locality constraint on the other receivers $\{A_j:j\neq i\}$ and regard them as a single receiver $A_{i^c}$ that receives channel output from all of $\{\mc{N}_j:j\neq i\}$. Clearly, the GHZ distribution capacity between $A_1,\cdots,A_N$ cannot be higher than the EPR distribution capacity between $A_i$ and $A_{i^c}$. This implies the following chain of inequalities:
    \begin{align}
        E(\mc{N}_1,\cdots,\mc{N}_N) \leq E\left(\mc{N}_i,\bigotimes\nolimits_{j\neq i}\mc{N}_j \right) \leq E(\mc{N}_i,\text{id}) \leq E_{cpp}(\mc{N}_i).
    \end{align}
    Since this argument works for arbitrary $i$, the stated theorem holds.
\end{proof}

\noindent A few remarks are in order here. First, perhaps surprisingly, the achievable rate we derived in Theorem \ref{thm:general-lower-bound} effectively neglects the less noisy channels, treating them as if they were noiseless. Second, this rate can be achieved without the assistance of classical communication. In fact, without classical communication, this is precisely the largest rate that is achievable. Finally, the upper bound is also tight when LOCC operations between the receivers and the central node are allowed. To see this, the source node can establish $E_{cpp}(\mc{N}_i)$ EPR pairs between $S$ and $A_i$ per use of each $\mc{N}_i$. Then the source node prepares an $N$-partite GHZ state and teleports each subsystem to $A_i$ by using the established EPR pairs. This establishes GHZ states among $A_i$ at a rate of $\min_i E_{cpp}(\mc{N}_i)$. Therefore, one can view the quantitative difference between the minimum quantum capacity $\min_i Q(\mc{N}_i)$, the GHZ distribution capacity $E(\mc{N}_1,\cdots,\mc{N}_N)$, and the minimum CPP-assisted entanglement generation capacity $\min_i E_{cpp}(\mc{N}_i)$ as a result of different levels of classical communication between parties in a quantum network.

A corollary of the previous theorem is an exact expression of the GHZ distribution capacity for when the ``worst'' channel is a dephasing channel.
\begin{proposition}\label{prop:dephasing}
    Consider $N$ channels $\mc{Z}_{p},\mc{N}_2,\cdots,\mc{N}_N$ where the $\mc{N}_i$ are channels satisfying $Q(\mc{N}_i)\geq Q(\mc{Z}_p)=1-h_2(p)$ for all $i>1$.  Then 
    $E(\mc{Z}_{p},\mc{N}_2,\cdots,\mc{N}_{N})=1-h_2(p)$.
\end{proposition}
\begin{proof}
    The statement follows directly from the fact that $Q(\mc{Z}_p)=E_{cpp}(\mc{Z}_p)=1-h_2(p)$ \cite{tomamichel-2017}.
\end{proof}

There is an alternative strategy to distribute GHZ states to $N$ parties that contrasts with Thm. \ref{thm:general-lower-bound}. The central node can distribute bipartite EPR pairs to every pair of receivers. The receivers then use these EPR pairs to build GHZ states. This is detailed in the following theorem.
\begin{theorem}
    $E(\mc{N}_1,\cdots,\mc{N}_N) \geq \displaystyle\frac{N}{2(N-1)}\min_{i \neq j}E(\mc{N}_i,\mc{N}_j)$.
    \label{thm:general-multipartite-lower-bound}
\end{theorem}
\begin{proof}
    For sufficiently large $n$, the central node can distribute $nE(\mc{N}_A,\mc{N}_B)$ EPR pairs to two receivers $A$ and $B$ by $n$ uses of $\mc{N}_A$ and $n$ uses of $\mc{N}_B$. By using $n(N-1)$ times each $\mc{N}_i$, each pair of receivers $A_i$ and $A_j$ will share $nE(\mc{N}_i,\mc{N}_j)$ EPR pairs. Let $\ket{\Phi_N}^{A_1\cdots A_N}\coloneqq\bigotimes_{i\neq j}\ket{\phi_+}^{A_iA_j}$, then the receivers share at least $\min_{i,j} nE(\mc{N}_i,\mc{N}_j)$ copies of $\ket{\Phi_N}$. Now, $\ket{\Phi_N}^{\otimes 2}$ can be converted to $\ket{\mathrm{GHZ}_N}^{\otimes N}$ via LOCC. This is because each $A_i$ can locally prepare an $N$-partite GHZ state and use the EPR pairs between $A_i$ and the other parties to create $\ket{\mathrm{GHZ}_N}$. Two EPR pairs between each $A_i$ and $A_j$ will be used (once when $A_i$ prepares the GHZ state and once when $A_j$ prepares the GHZ state). Therefore, a rate of $\frac{N}{2(N-1)}\min_{i \neq j}E(\mc{N}_i,\mc{N}_j)$ is achievable. 
\end{proof}

\section{Conclusion and Discussion}\label{sect:conclusion}
In conclusion, we study the entanglement distribution scheme where a single central source aims to establish entanglement among $N$ receivers through quantum channels $\mc{N}_1,\cdots,\mc{N}_N$. We define the EPR distribution capacity $E(\mc{N}_1,\mc{N}_2)$ and the GHZ distribution capacity $E(\mc{N}_1,\cdots,\mc{N}_N)$ as the largest asymptotic rate of establishing EPR states among two receivers and $N$-partite GHZ states among the $N$ receivers, respectively. We show that EPR distribution is tightly related to the problem of assisted entanglement distillation. Specifically, we prove a lower bound in terms of the asymptotic entanglement of assistance and show that this bound is tight for teleportation-covariant channels. This result allows us to exactly characterize the EPR distribution capacity of two erasure channels. Then, we give alternative simple lower bounds for the EPR distribution capacity over generic channels as well as the EPR distribution capacity over two generalized amplitude damping channels. Finally, we bound the GHZ-distribution capacity in terms of the minimum quantum capacity and minimum CPP-assisted entanglement generation capacities of the channels. Applying these bounds, we exactly determine the GHZ distribution capacity when the most noisy channel is a dephasing channel.

Our work sheds light on the novel connection between two distinct information theoretic tasks and motivates a deeper understanding of assisted entanglement distillation. Moreover, our work also hints at a better scheme of distributing entanglement. Suppose two parties $A$ and $B$ are connected by an optical fiber, and we assume that the polarization dispersion in the optical fiber are characterized by dephasing channels \cite{liu-2022}. If they wish to share an entangled pair, instead of generating entanglement at one of the parties, say $A$, who sends half of the entangled state through the entire length of the fiber, a better strategy is to generate entanglement from a source $S$ that is located centrally between $A$ and $B$. By virtue of Proposition \ref{prop:dephasing}, in this alternative strategy, the rate of entanglement distribution depends only on the more noisy link between the source and the receivers, treating the other link as noiseless. The analysis of our entanglement distribution scheme shows how it is possible to quantify its performance in terms of information rates. Its operational relevance and the existence of alternative schemes for different scenarios, therefore, make this area of research rich of insights and practical applications which are worth further exploring.

\section*{Acknowledgments}
We are grateful to Koji Azuma for the valuable discussions. This work is supported by the Department of Energy (DOE) Q-NEXT: National Quantum Information Science Research Center.

\bibliographystyle{unsrtnat}
\bibliography{refs}

\end{document}